\documentclass [11pt,a4paper] {article}

\usepackage[cp1252]{inputenc}

\usepackage{amssymb}
\usepackage{amsmath}
\usepackage{amsfonts,amssymb}
\usepackage[dvips]{graphicx}
\usepackage{bbm}
\usepackage{enumerate}

\usepackage{amsthm}

\DeclareMathAlphabet{\mathpzc}{OT1}{pzc}{m}{it}

\def\SmallColSep{\setlength{\arraycolsep}{1pt}}

\newtheorem{observation}{Observation} %Define Observation
\newtheorem{corollary}{Corollary} %Define Corollary

\begin{document}

\title{Contextuality and the fundamental theorem of noncommutative algebra}

\author{Arkady Bolotin\footnote{$Email: arkadyv@bgu.ac.il$\vspace{5pt}} \\ \textit{Ben-Gurion University of the Negev, Beersheba (Israel)}}

\maketitle

\begin{abstract}\noindent In the paper it is shown that the Kochen-Specker theorem follows from Burnside's theorem on noncommutative algebras. Accordingly, contextuality (as an impossibility of assigning binary values to projection operators independently of their contexts) is merely an inference from Burnside's fundamental theorem of the algebra of linear transformations on a Hilbert space of finite dimension.\\

\end{abstract}

\section{Introduction}  %{I}

\noindent Consider the set of the propositions $\{\diamond\}$ where the symbol $\diamond$ stands for any proposition, compound or simple.\\

\noindent Let $v_C$ be a truth-value assignment function that denotes \textit{a truth valuation in a circumstance $C$}, that is, a mapping from the set $\{\diamond\}$ to the set of truth-values $\{\mathfrak{v}_i\}_{i=1}^N$ (where $N \ge 2$) relative to a circumstance of evaluation indicated by $C$:\smallskip

\begin{equation} \label{1} %{Eq.1}
   v_C:
   \{\diamond\}
   \rightarrow
   \{\mathfrak{v}_i\}_{i=1}^N
   \;\;\;\;  .
\end{equation}
\smallskip

\noindent Commonly (see, e.g., \cite{Dalen}), the image of $\diamond$ under $v_C$ is written using the double-bracket notation, namely,\smallskip

\begin{equation} \label{2} %{Eq.2}
   v_C(\diamond)
   =
   {[\![ \diamond ]\!]}_C
   \;\;\;\;  .
\end{equation}
\smallskip

\noindent The truth-value assignment function $v_C$ expresses the notion of \textit{not-yet-verified truth values}: It specifies in advance the truth-value obtained from the verification of a proposition.\\

\noindent To relate the set $\{\diamond\}$ to the states of a physical system, one can employ \textit{a predicate} -- i.e., a statement whose truth value depends on the values of its variables. For example, in the case of a system associated with the classical phase space $\Gamma$, the predicate $P$ on $\Gamma$ can be defined as a function from the phase space to the set of truth-values:\smallskip

\begin{equation} \label{3} %{Eq.3}
   P:
   \Gamma
   \rightarrow
   \{\mathfrak{v}_i\}_{i=1}^N
   \;\;\;\;  .
\end{equation}
\smallskip

\noindent Let $\{P_\diamond\}$ denote the set of the predicates uniquely (i.e., one-to-one) connected to the set of the propositions $\{\diamond\}$. Then, one can introduce the valuation equivalence\smallskip

\begin{equation} \label{4} %{Eq.4}
   v_{\gamma}(P_\diamond)
   =
   {[\![ \diamond ]\!]}_{\gamma}
   \;\;\;\;  ,
\end{equation}
\smallskip

\noindent which signifies that the truth-value of the proposition $\diamond$ in the state $\gamma \in \Gamma$ is equated with the value of the corresponding predicate $P_{\diamond}$ obtained in this state, i.e., $v_{\gamma}(P_{\diamond}) = P_{\diamond}(\gamma)$.\\

\noindent Provided that in the case of a classical system a predicate is just an indicator function that only takes the values 0 or 1 (where 0 denotes the falsity and 1 denotes the truth), explicitly,\smallskip

\begin{equation} \label{5} %{Eq.5}
   v_{\gamma}(P_\diamond)
   =
   \left\{
      \begin{array}{l}
         1, \;\;\;\;  \gamma \in U_{\diamond}\\
         0, \;\;\;\;  \gamma \in V_{\diamond}
      \end{array}
   \right.
   \;\;\;\;  ,
\end{equation}
\smallskip

\noindent in which $U_{\diamond}$ and $V_{\diamond}$ are some linear subspaces of $\Gamma$ such that $\Gamma = U_{\diamond} \oplus V_{\diamond}$, the relation between the set of the predicates and the set of truth-values is a bivaluation:\smallskip

\begin{equation} \label{6} %{Eq.6}
   v_{\gamma}:
   \{P_\diamond\}
   \rightarrow
   \{0,1\}
   \;\;\;\;  .
\end{equation}
\smallskip

\noindent Accordingly, the elements $\gamma$ of the classical phase space $\Gamma$ represent categorical properties that the classical system possesses or does not. What is more, the bivaluation relation (\ref{6}) is \textit{a total function}. This means that any proposition related to a classical system obeys \textit{the principle of bivalence} (asserting that a proposition can be either true or false \cite{Beziau}).\\

\noindent To define the truth-value assignment for a quantum system associated with a Hilbert space $\mathcal{H}$, one can assume the valuation equivalence analogous to (\ref{4})\smallskip

\begin{equation} \label{7} %{Eq.7}
   v_{|\Psi\rangle}(\hat{P}_{\diamond})
   =
   {[\![ \diamond ]\!]}_{|\Psi\rangle}
   \;\;\;\;  ,
\end{equation}
\smallskip

\noindent where $\hat{P}_{\diamond}$ denotes a projection operator on $\mathcal{H}$ uniquely connected with a proposition $\diamond$, while $|\Psi\rangle$ stands for a vector in $\mathcal{H}$ describing system's state. In line with this equivalence, the truth value of the proposition $\diamond$ in the state $|\Psi\rangle$ is equated with the value of the corresponding projection operator $\hat{P}_{\diamond}$ obtained in this state.\\

\noindent As it can be readily seen, the difference between the equivalence (\ref{7}) and its classical counterpart (\ref{4}) is not only one that in the former the argument of the value assignment function is an operator on a Hilbert space $\mathcal{H}$ (instead of a predicate on the classical phase space $\Gamma$ in the latter) but also (and more importantly) one that the relation between the set $\{\hat{P}_{\diamond}\}$ and the set $\{0,1\}$, namely,\smallskip

\begin{equation} \label{8} %{Eq.8}
   v_{|\Psi\rangle}:
   \{\hat{P}_{\diamond}\}
   \rightarrow
   \{0,1\}
   \;\;\;\;  ,
\end{equation}
\smallskip

\noindent cannot be a total function in accordance with the Kochen-Specker theorem \cite{Kochen, Peres}. This means that at least one proposition related to a quantum system does not obey the principle of bivalence: The said proposition may have a truth-value different from 0 and 1 (as it is argued in \cite{Pykacz95, Pykacz15}) or no truth-value at all (in line with the supervaluation approach suggested in \cite{Bolotin}).\\

\noindent Assuming that the quantum value assignment function $v_{|\Psi\rangle}$ can be presented as an indicator function similarly to the case of a classical system, that is,\smallskip

\begin{equation} \label{9} %{Eq.9}
   v_{|\Psi\rangle}(\hat{P}_{\diamond})
   =
   \left\{
      \begin{array}{l}
         1, \;\;\;\;  |\Psi\rangle \in U_{\diamond}\\
         0, \;\;\;\;  |\Psi\rangle \in V_{\diamond}
      \end{array}
   \right.
   \;\;\;\;  ,
\end{equation}
\smallskip

\noindent where $U_{\diamond}$ and $V_{\diamond}$ are some linear subspaces in $\mathcal{H}$ such that $\mathcal{H} = U_{\diamond} \oplus V_{\diamond}$, the question is, what algebraic properties of the linear subspaces in $\mathcal{H}$ cause the failure of the principle of bivalence? Correspondingly, can the Kochen-Specker theorem be derived from the algebra over $\mathcal{H}$?\\

\noindent Let us answer these questions in the presented paper.\\

\section{Invariant subspaces for projection operators}  %{II}

\noindent Recall the following definitions. \textit{The column space} (a.k.a. \textit{range}), $\mathrm{ran}(\hat{P})$, of the projection (i.e., self-adjoint and idempotent) operator $\hat{P}$ is the subset of the vectors $|\Psi\rangle \in \mathcal{H}$ that are in the image of $\hat{P}$, namely,\smallskip

\begin{equation} \label{10} %{Eq.10}
   \mathrm{ran}(\hat{P})
   =
   \left\{
      |\Psi\rangle \in \mathcal{H}
      :
      \;
      \hat{P}|\Psi\rangle = |\Psi\rangle
   \right\}
   \;\;\;\;  .
\end{equation}
\smallskip

\noindent Likewise, \textit{the null space} (a.k.a. \textit{kernel}), $\mathrm{ker}(\hat{P})$, of the projection operator $\hat{P}$ is the subset of the vectors $|\Psi\rangle \in \mathcal{H}$ that are mapped to zero by $\hat{P}$, namely,\smallskip

\begin{equation} \label{11} %{Eq.11}
   \mathrm{ker}(\hat{P})
   =
   \left\{
      |\Psi\rangle \in \mathcal{H}
      :
      \;
      \hat{P}|\Psi\rangle = 0
   \right\}
   \;\;\;\;  .
\end{equation}
\smallskip

\noindent Thus, any $\hat{P}$ is the identity operator $\hat{1}$ on $\mathrm{ran}(\hat{P})$ and the zero operator $\hat{0}$ on $\mathrm{ker}(\hat{P})$.\\

\noindent The column and null spaces are complementary in the same way as $\hat{P}$ and $\hat{1} - \hat{P}$, that is,\smallskip

\begin{equation} \label{12} %{Eq.12}
   \mathrm{ran}(\hat{P})
   =
   \mathrm{ker}(\hat{1} - \hat{P})
   \;\;\;\;  ,
\end{equation}

\begin{equation} \label{13} %{Eq.13}
   \mathrm{ker}(\hat{P})
   =
   \mathrm{ran}(\hat{1} - \hat{P})
   \;\;\;\;  .
\end{equation}
\smallskip

\noindent Moreover, they produce the direct sum\smallskip

\begin{equation} \label{14} %{Eq.14}
   \mathrm{ran}(\hat{P})
   \oplus
   \mathrm{ran}(\hat{P})
   =
   \mathrm{ran}(\hat{1})
   =
   \mathcal{H}
   \;\;\;\;  ,
\end{equation}
\smallskip

\noindent and they are orthogonal to each other:\smallskip

\begin{equation} \label{15} %{Eq.15}
   \mathrm{ran}(\hat{P})
   \cap
   \mathrm{ker}(\hat{P})
   =
   \mathrm{ran}(\hat{0})
   =
   \{ 0\}
   \;\;\;\;  ,
\end{equation}
\smallskip

\noindent where 0 denotes the zero vector in any vector space and $\{ 0\}$ stands for the zero subspace. Thus, $\mathrm{ker}(\hat{P})$ is the orthogonal complement of $\mathrm{ran}(\hat{P})$, and vice versa.\\

\noindent Also recall that a subspace $U \subseteq \mathcal{H}$ is called \textit{an invariant subspace under} $\hat{P}$ if\smallskip

\begin{equation} \label{16} %{Eq.16}
   |\Psi\rangle \in U
   \implies
   \hat{P}|\Psi\rangle \in U
   \;\;\;\;  ,
\end{equation}
\smallskip

\noindent that is, $\hat{P}(U)$ is contained in $U$ and so\smallskip

\begin{equation} \label{17} %{Eq.17}
   \hat{P}:
   U
   \rightarrow
   U
   \;\;\;\;  .
\end{equation}
\smallskip

\noindent Obviously, the space $\mathcal{H}$ itself as well as the zero subspace $\{0\}$ are \textit{trivially invariant subspaces} for any projection operator $\hat{P}$.\smallskip

\begin{observation}
For each projection operator $\hat{P}$ there are two nontrivial invariant subspaces, namely, $\mathrm{ran}(\hat{P})$ and $\mathrm{ker}(\hat{P})$.
\end{observation}

\begin{proof}
To see this, let $|\Psi\rangle \in \mathrm{ran}(\hat{P})$. Since $\hat{P}|\Psi\rangle = |\Psi\rangle$ one gets $\hat{P}|\Psi\rangle \in \mathrm{ran}(\hat{P})$, and so $\hat{P}: \mathrm{ran}(\hat{P}) \rightarrow \mathrm{ran}(\hat{P})$. Similarly, let $|\Psi\rangle \in \mathrm{ker}(\hat{P})$. This means that $\hat{P}|\Psi\rangle = 0$. On the other hand, $0 \in \mathrm{ker}(\hat{P})$, which implies $\hat{P}: \mathrm{ker}(\hat{P}) \rightarrow \mathrm{ker}(\hat{P})$.
\end{proof}
\smallskip

\noindent The presence of two invariant subspaces for each projection operator $\hat{P}$ motivates the definition of the valuation $v_{|\Psi\rangle}$ as \textit{a bivalent function}, that is,\smallskip 

\begin{equation} \label{18} %{Eq.18}
   v_{|\Psi\rangle}(\hat{P})
   =
   \left\{
      \begin{array}{l}
         1, \;\;\;\;  |\Psi\rangle \in \mathrm{ran}(\hat{P})\\
         0, \;\;\;\;  |\Psi\rangle \in \mathrm{ker}(\hat{P})
      \end{array}
   \right.
   \;\;\;\;  .
\end{equation}
\smallskip

\noindent Take the identity operator $\hat{1}$. Given $\mathrm{ran}(\hat{1}) = \mathcal{H}$ and $\mathrm{ker}(\hat{1}) = \{0\}$, from the said definition it follows that for any $|\Psi\rangle \in \mathcal{H}$ and $|\Psi\rangle \neq 0$, i.e., for any admissible state $|\Psi\rangle$ of a system, $v_{|\Psi\rangle}(\hat{1}) = 1$. This indicates that the identity operator $\hat{1}$ relates to \textit{a tautology} $\top$ (i.e., a proposition that is true in any admissible state of the system), namely, $v_{|\Psi\rangle}(\hat{1}) = {[\![ \top ]\!]}_{|\Psi\rangle} = 1$.\\

\noindent For zero operator $\hat{0}$, one gets in accordance with (\ref{18}) that in any admissible state of the system, $v_{|\Psi\rangle}(\hat{0}) = 0$. This implies that the zero operator $\hat{0}$ relates to \textit{a contradiction} $\bot$ (i.e., a proposition that is false in any admissible state of the system): $v_{|\Psi\rangle}(\hat{0}) = {[\![ \bot ]\!]}_{|\Psi\rangle} = 0$.\\

\section{Burnside’s theorem on the noncommutative algebra}  %{III}

\noindent Let $L(\mathcal{H})$ denote the algebra of linear transformations on $\mathcal{H}$ and let $\Sigma$ represents the collection of projection operators on $\mathcal{H}$. Consider a nonempty subset $\Sigma^{(q)} \subset \Sigma$ comprising projection operators $\hat{P}_i^{(q)}$ that meet the conditions\smallskip

\begin{equation} \label{19} %{Eq.19}
   \hat{P}_i^{(q)} \hat{P}_j^{(q)}
   =
   \hat{P}_j^{(q)} \hat{P}_i^{(q)}
   =
   \hat{0}
   \;\;\;\;  ,
\end{equation}
\smallskip

\noindent where $i \neq j$, and\smallskip

\begin{equation} \label{20} %{Eq.20}
   \sum_{\hat{P}_i^{(q)} \in \, \Sigma^{(q)}}
      \hat{P}_i^{(q)}
   =
   \hat{1}
   \;\;\;\;  .
\end{equation}
\smallskip

\noindent Such a subset $\Sigma^{(q)}$ is said to be \textit{a maximal} (a.k.a. \textit{complete}) \textit{context}.\\

\noindent Let $\mathrm{Lat}(\hat{P}_i^{(q)})$ be the family of subspaces invariant under the projection operator $\hat{P}_i^{(q)}$, namely,\smallskip

\begin{equation} \label{21} %{Eq.21}
   \mathrm{Lat}(\hat{P}_i^{(q)})
   =
   \left\{
      \mathrm{ran}(\hat{0})
      ,\,
      \mathrm{ran}(\hat{P}_i^{(q)})
      ,\,
      \mathrm{ker}(\hat{P}_i^{(q)})
      ,\,
      \mathrm{ran}(\hat{1})
   \right\}
   \;\;\;\;  ,
\end{equation}
\smallskip

\noindent such that $\mathrm{Lat}(\hat{P}_i^{(q)})$ forms \textit{a lattice}: The operation meet $\sqcap$ of this lattice corresponds to the interception $\mathcal{Q} \cap \mathcal{W}$ and the lattice operation join $\sqcup$ corresponds to the smallest closed subspace of $\mathrm{Lat}(\hat{P}_i^{(q)})$ containing the union $\mathcal{Q} \cup \mathcal{W}$, where $\mathcal{Q} \neq \mathcal{W}$ and $\mathcal{Q}, \mathcal{W} \in \mathrm{Lat}(\hat{P}_i^{(q)})$. This lattice is bounded, i.e., it has the greatest element $\mathrm{ran}(\hat{1}) = \mathcal{H}$ and the least element $\mathrm{ran}(\hat{0}) = \{0\}$.\\

\noindent Now, consider the invariant subspaces invariant under each projection operator $\hat{P}_i^{(q)}$ in the maximal context $\Sigma^{(q)}$:\smallskip

\begin{equation} \label{22} %{Eq.22}
   \mathrm{Lat}(\Sigma^{(q)})
   =
   \bigcap_{\hat{P}_i^{(q)} \in \, \Sigma^{(q)}}
      \mathrm{Lat}(\hat{P}_i^{(q)})
   \;\;\;\;  .
\end{equation}
\smallskip

\noindent Given that for any maximal context the following interceptions hold\smallskip

\begin{equation} \label{23} %{Eq.23}
   \mathrm{ran}(\hat{P}_i^{(q)})
   \cap
   \mathrm{ran}(\hat{P}_j^{(q)})
   =
   \mathrm{ran}(\hat{P}_i^{(q)}\hat{P}_j^{(q)})
   =
   \{ 0\}
   \;\;\;\;  ,
\end{equation}

\begin{equation} \label{24} %{Eq.24}
   \mathrm{ran}(\hat{P}_i^{(q)})
   \cap
   \mathrm{ker}(\hat{P}_j^{(q)})
   =
   \mathrm{ran}(\hat{P}_i^{(q)})
   \cap
   \mathrm{ran}(\hat{1} - \hat{P}_j^{(q)})
   =
   \mathrm{ran}(\hat{P}_i^{(q)})
   \;\;\;\;  ,
\end{equation}

\begin{equation} \label{25} %{Eq.25}
   \mathrm{ker}(\hat{P}_i^{(q)})
   \cap
   \mathrm{ran}(\hat{P}_j^{(q)})
   =
   \mathrm{ran}(\hat{1} - \hat{P}_i^{(q)})
   \cap
   \mathrm{ker}(\hat{P}_j^{(q)})
   =
   \mathrm{ran}(\hat{P}_j^{(q)})
   \;\;\;\;  ,
\end{equation}

\begin{equation} \label{26} %{Eq.26}
   \mathrm{ker}(\hat{P}_i^{(q)})
   \cap
   \mathrm{ker}(\hat{P}_j^{(q)})
   =
   \mathrm{ran}(\hat{1} - \hat{P}_i^{(q)} - \hat{P}_j^{(q)})
   =
   \mathrm{ran}
      \bigg(
            \sum_{k \neq i,j}
               \hat{P}_k^{(q)}\!
      \bigg)\!
   \;\;\;\;  ,
\end{equation}
\smallskip

\noindent the interception $\mathrm{Lat}(\Sigma^{(q)})$ can be presented as\smallskip

\begin{equation} \label{27} %{Eq.27}
   \mathrm{Lat}(\Sigma^{(q)})
   =
   \bigg\{
      \mathrm{ran}(\hat{0})
      ,
      \left\{
         \mathrm{ran}(\hat{P}_i^{(q)})
         ,
         \mathrm{ker}(\hat{P}_i^{(q)})
      \right\}_{i=1}
      ,
      \mathcal{R}^{(q)}
      ,
      \mathrm{ran}(\hat{1})
   \bigg\}
   \;\;\;\;  ,
\end{equation}
\smallskip

\noindent where $\mathcal{R}^{(q)}$ stands for\smallskip

\begin{equation} \label{28} %{Eq.28}
   \mathcal{R}^{(q)}
   =
   \Bigg\{
   \mathrm{ran}
      \bigg(
            \sum_{k \neq i,j}
               \hat{P}_k^{(q)}\!
      \bigg)
      ,\,
   \mathrm{ran}
      \bigg(
            \sum_{j \neq i,k}
               \hat{P}_j^{(q)}\!
      \bigg)
   ,\,
   \mathrm{ran}
      \bigg(
            \sum_{i \neq j,k}
               \hat{P}_i^{(q)}\!
      \bigg)
   ,\,
   \mathrm{ran}
      \bigg(
            \sum_{l \neq i,j,k}
               \hat{P}_l^{(q)}\!
      \bigg)
   ,\,
   \dots   
   \Bigg\}
   \;\;\;\;  .
\end{equation}
\smallskip

\noindent Suppose that the system is prepared in a pure state $|\Psi\rangle \in \mathrm{ran}(\hat{P}_i^{(q)})$. Then, according to (\ref{18}), $v_{|\Psi\rangle}(\hat{P}_i^{(q)}) = 1$. As a result of (\ref{23}) (expressing that all $\mathrm{ran}(\hat{P}_i^{(q)})$ are orthogonal to each other), the vector $|\Psi\rangle$ also resides in the null space of any other projection operator in the maximal context $\Sigma^{(q)}$, i.e., $|\Psi\rangle \in \mathrm{ker}(\hat{P}_j^{(q)})$, which gives $v_{|\Psi\rangle}(\hat{P}_j^{(q)}) = 0$. Hence, in the maximal context $\Sigma^{(q)}$ only one projection operator can be assigned the value 1, and so\smallskip

\begin{equation} \label{29} %{Eq.29}
   v_{|\Psi\rangle}
   \bigg(
   \sum_{\hat{P}_i^{(q)} \in \, \Sigma^{(q)}}
      \hat{P}_i^{(q)}\!
   \bigg)
   =
   \sum_{\hat{P}_i^{(q)} \in \, \Sigma^{(q)}}
      v_{|\Psi\rangle}\!\!
      \left(
         \hat{P}_i^{(q)}
      \right)
   =
   1
   \;\;\;\;  .
\end{equation}
\smallskip

\noindent Consider the invariant subspaces invariant under each maximal context $\Sigma^{(q)} \subset \Sigma$, that is,\smallskip

\begin{equation} \label{30} %{Eq.30}
   \mathrm{Lat}(\Sigma)
   =
   \bigcap_{\Sigma^{(q)} \subset \Sigma}
      \mathrm{Lat}(\Sigma^{(q)})
   \;\;\;\;  .
\end{equation}

\begin{observation}
Where $\mathrm{Lat}(\Sigma)$ to contain some nontrivial invariant subspace(s), a logic defined as the relations between projection operators in $\Sigma$ would have a bivalent semantics.
\end{observation}

\begin{proof}
Suppose that $\mathrm{Lat}(\Sigma)$ contains a nontrivial invariant subspace $\mathrm{ran}(\hat{P}_i^{(q)})$. Since $\mathrm{Lat}(\Sigma)$ is the intersection of all the lattices $\mathrm{Lat}(\Sigma^{(q)})$, this would mean that $\mathrm{ran}(\hat{P}_i^{(q)})$ is the member of each $\mathrm{Lat}(\Sigma^{(q)})$ and thus orthogonal to all other column spaces in each $\Sigma^{(q)}$. In the case where $v_{|\Psi\rangle}(\hat{P}_i^{(q)}) = 1$, all other truth values of the projection operators in $\Sigma$ would be zero, which would produce $\sum_{i} v_{|\Psi\rangle}(\hat{P}_i^{(q)}) = 1$ for each $\Sigma^{(q)}$. In this way, all the projection operators in $\Sigma$ would obey the principle of bivalence.
\end{proof}
\smallskip

\noindent Assume $\Sigma = L(\mathcal{H})$, i.e., the collection $\Sigma$ includes all the projection operators on system’s Hilbert space $\mathcal{H}$. If $\mathcal{H}$ is finite-dimensional (and $\mathrm{dim}(\mathcal{H})$ is greater than 1), then, according to \textit{Burnside's Theorem} \cite{Burnside, Rosenthal, Lomonosov}, $\mathrm{Lat}(\Sigma)$ is \textit{irreducible}, i.e., has no nontrivial invariant subspace:\smallskip

\begin{equation} \label{31} %{Eq.31}
   \Sigma
   =
   L(\mathcal{H})
   \;
   \implies
   \;
   \mathrm{Lat}(\Sigma)
   =
   \left\{
      \mathrm{ran}(\hat{0})
      ,
      \mathrm{ran}(\hat{1})
   \right\}
   \;\;\;\;  .
\end{equation}

\begin{observation}
This means that for the given system the principle of bivalence fails.
\end{observation}

\begin{proof}
Suppose that $\sum_{i} v_{|\Psi\rangle}(\hat{P}_i^{(q)}) = 1$ for the certain $\Sigma^{(q)}$. Because it is irreducible, $\mathrm{Lat}(\Sigma)$ does not have any $\mathrm{ran}(\hat{P}_i^{(q)}) \in \mathrm{Lat}(\Sigma^{(q)})$. So, at least one nontrivial invariant subspace, say, $\mathrm{ran}(\hat{P}_k^{(w)}) \in \mathrm{Lat}(\Sigma^{(w)})$, where $w \neq q$, is not orthogonal to $\mathrm{ran}(\hat{P}_i^{(q)}) \in \mathrm{Lat}(\Sigma^{(q)})$. In consequence, the proposition associated with $\hat{P}_k^{(w)}$ cannot be bivalent alongside the propositions connected with $\hat{P}_i^{(q)} \in \Sigma^{(q)}$, i.e., $v_{|\Psi\rangle}(\hat{P}_k^{(w)}) \neq \{0,1\}$.
\end{proof}
\smallskip

\begin{corollary}
There is a collection of the projection operators relating to a system with a finite-dimensional Hilbert space, namely, $\Sigma^{\prime} \subset L(\mathcal{H})$, such that $\mathrm{Lat}(\Sigma^{\prime})$ contains no nontrivial elements.
\end{corollary}

\begin{proof}
This follows directly from the version of Burnside's Theorem presented in (\ref{31}). Truly, since for the said system $\mathrm{Lat}(\Sigma)$ is irreducible, there must exist maximal contexts $\Sigma^{(q)}$ and $\Sigma^{(w)}$ whose nontrivial column spaces $\mathrm{ran}(\hat{P}_i^{(q)}) \in \mathrm{Lat}(\Sigma^{(q)})$ and $\mathrm{ran}(\hat{P}_k^{(w)}) \in \mathrm{Lat}(\Sigma^{(w)})$ are not orthogonal to each other. Correspondingly, the family of subspaces $\mathrm{Lat}(\Sigma^{\prime})$ invariant under both $\Sigma^{(q)}$ and $\Sigma^{(w)}$ would have no nontrivial invariant subspace.
\end{proof}
\smallskip

\section{Algebra of matrices on $\mathbb{C}^2$}  %{IV}

\noindent Let us demonstrate the application of Burnside’s Theorem to the set of the projection operators on the two dimensional Hilbert space $\mathcal{H} = \mathbb{C}^2$. Comprised of the eigenvectors of the Pauli matrices $\sigma_z$, $\sigma_x$ and $\sigma_y$ these projection operators are\smallskip

\begin{equation} \label{32} %{Eq.32}
   \hat{P}_{1}^{(z)}
   =
   \!\left[
      \begingroup\SmallColSep
      \begin{array}{r r}
         1 & 0 \\
         0 & 0
      \end{array}
      \endgroup
   \right]
   \,
   ,
   \,\,
   \hat{P}_{2}^{(z)}
   =
   \!\left[
      \begingroup\SmallColSep
      \begin{array}{r r}
         0 & 0 \\
         0 & 1
      \end{array}
      \endgroup
   \right]
   \;\;\;\;  ,
\end{equation}

\begin{equation} \label{33} %{Eq.33}
   \hat{P}_{1}^{(x)}
   =
   \!\frac{1}{2}
   \!\left[
      \begingroup\SmallColSep
      \begin{array}{r r}
         1 & 1 \\
         1 & 1
      \end{array}
      \endgroup
   \right]
   \,
   ,
   \,\,
   \hat{P}_{2}^{(x)}
   =
   \!\frac{1}{2}
   \!\left[
      \begingroup\SmallColSep
      \begin{array}{r r}
         1 & -1 \\
        -1 &  1
      \end{array}
      \endgroup
   \right]
   \;\;\;\;  ,
\end{equation}

\begin{equation} \label{34} %{Eq.34}
   \hat{P}_{1}^{(y)}
   =
   \!\frac{1}{2}
   \!\left[
      \begingroup\SmallColSep
      \begin{array}{r r}
         1 & -i \\
          i &  1
      \end{array}
      \endgroup
   \right]
   \,
   ,
   \,\,
   \hat{P}_{2}^{(y)}
   =
   \!\frac{1}{2}
   \!\left[
      \begingroup\SmallColSep
      \begin{array}{r r}
         1 & i \\
        -i  & 1
      \end{array}
      \endgroup
   \right]
   \;\;\;\;  .
\end{equation}
\smallskip

\noindent Since $\hat{P}_{1}^{(q)} \hat{P}_{2}^{(q)} = \hat{P}_{2}^{(q)} \hat{P}_{1}^{(q)} = \hat{0}$ and $\hat{P}_{1}^{(q)} + \hat{P}_{2}^{(q)} = \hat{1}$, they make up three maximal contexts $\Sigma^{(q)}$, namely,\smallskip

\begin{equation} \label{35} %{Eq.35}
   \Sigma^{(q)}
   =
   \left\{
      \hat{P}_{i}^{(q)}
   \right\}_{i = 1}^2
   \subset
   L(\mathbb{C}^2)
   \;\;\;\;  ,
\end{equation}
\smallskip

\noindent where $L(\mathbb{C}^2)$ denotes the collection of all linear transformations $\mathbb{C}^2 \rightarrow \mathbb{C}^2$ (i.e., the algebra over $\mathbb{C}^2$).\\

\noindent The invariant subspaces $ \mathrm{Lat}(\Sigma^{(q)})$ invariant under each $\hat{P}_{i}^{(q)} \in \Sigma^{(q)}$ take the form\smallskip

\begin{equation} \label{36} %{Eq.36}
   \mathrm{Lat}(\Sigma^{(z)})
   =
   \bigg\{
      \{0\}
      ,
      \left\{
         \!\left[
            \begingroup\SmallColSep
            \begin{array}{r}
               a \\
               0
            \end{array}
            \endgroup
         \right]\!
      \right\}
      ,
      \left\{
         \!\left[
            \begingroup\SmallColSep
            \begin{array}{r}
               0 \\
               a
            \end{array}
            \endgroup
         \right]\!
      \right\}
      ,
      \left\{
         \!\left[
            \begingroup\SmallColSep
            \begin{array}{r}
               0 \\
               b
            \end{array}
            \endgroup
         \right]\!
      \right\}
      ,
      \left\{
         \!\left[
            \begingroup\SmallColSep
            \begin{array}{r}
               b \\
               0
            \end{array}
            \endgroup
         \right]\!
      \right\}
      ,
      \mathbb{C}^2
   \bigg\}
   \;\;\;\;  ,
\end{equation}

\begin{equation} \label{37} %{Eq.37}
   \mathrm{Lat}(\Sigma^{(x)})
   =
   \bigg\{
      \{0\}
      ,
      \left\{
         \!\left[
            \begingroup\SmallColSep
            \begin{array}{r}
               a \\
               a
            \end{array}
            \endgroup
         \right]\!
      \right\}
      ,
      \left\{
         \!\left[
            \begingroup\SmallColSep
            \begin{array}{r}
                a \\
               -a
            \end{array}
            \endgroup
         \right]\!
      \right\}
      ,
      \left\{
         \!\left[
            \begingroup\SmallColSep
            \begin{array}{r}
                b \\
               -b
            \end{array}
            \endgroup
         \right]\!
      \right\}
      ,
      \left\{
         \!\left[
            \begingroup\SmallColSep
            \begin{array}{r}
                b \\
                b
            \end{array}
            \endgroup
         \right]\!
      \right\}
      ,
      \mathbb{C}^2
   \bigg\}
   \;\;\;\;  ,
\end{equation}

\begin{equation} \label{38} %{Eq.38}
   \mathrm{Lat}(\Sigma^{(z)})
   =
   \bigg\{
      \{0\}
      ,
      \left\{
         \!\left[
            \begingroup\SmallColSep
            \begin{array}{r}
               ia \\
                a
            \end{array}
            \endgroup
         \right]\!
      \right\}
      ,
      \left\{
         \!\left[
            \begingroup\SmallColSep
            \begin{array}{r}
                a \\
               ia
            \end{array}
            \endgroup
         \right]\!
      \right\}
      ,
      \left\{
         \!\left[
            \begingroup\SmallColSep
            \begin{array}{r}
                b \\
               ib
            \end{array}
            \endgroup
         \right]\!
      \right\}
      ,
      \left\{
         \!\left[
            \begingroup\SmallColSep
            \begin{array}{r}
               ib \\
                b
            \end{array}
            \endgroup
         \right]\!
      \right\}
      ,
      \mathbb{C}^2
   \bigg\}
   \;\;\;\;  ,
\end{equation}
\smallskip

\noindent where $a,b \in \mathbb{R}$.\\

\noindent Within each maximal context $\mathrm{Lat}(\Sigma^{(q)})$ the corresponding projection operators $\hat{P}_{i}^{(q)}$ are bivalent. For example, suppose the system is prepared in the state $|\Psi\rangle = [1,0]^{\mathrm{T}}$, then\smallskip

\begin{equation} \label{39} %{Eq.39}
   |\Psi\rangle
   =
   \!\left[
      \begingroup\SmallColSep
      \begin{array}{r}
          1 \\
          0
      \end{array}
      \endgroup
   \right]\!
   \in
   \left\{
      \begin{array}{l}
         \mathrm{ran}(\hat{P}_{1}^{(z)})
         =
         \left\{
            \!\left[
               \begingroup\SmallColSep
               \begin{array}{r}
                   a \\
                   0
               \end{array}
               \endgroup
            \right]\!
         \right\}
         \,
         \implies
         \,\,
         v_{|\Psi\rangle}(\hat{P}_{1}^{(z)})
         =
         1
         \\ % down
         \mathrm{ker}(\hat{P}_{2}^{(z)})
         =
         \left\{
            \!\left[
               \begingroup\SmallColSep
               \begin{array}{r}
                   b \\
                   0
               \end{array}
               \endgroup
            \right]\!
         \right\}
         \,
         \implies
         \,\,
         v_{|\Psi\rangle}(\hat{P}_{2}^{(z)})
         =
         0
      \end{array}
   \right.
   \;\;\;\;  .
\end{equation}
\smallskip

\noindent Since the Pauli matrices $\sigma_q$ form an orthogonal basis for the space $\mathbb{C}^2$, any matrix $\mathrm{M}^{2 \times 2} \in \mathbb{C}^2$ can be expressed as\smallskip

\begin{equation} \label{40} %{Eq.40}
   \mathrm{M}^{2 \times 2}
   =
   w\mathrm{I}^{2 \times 2}
   +
   \sum_{q = 1}^{3}
      u_q \sigma_q
   \;\;\;\;  ,
\end{equation}
\smallskip

\noindent where $w$ and $u_q$ are complex numbers, and $\mathrm{I}^{2 \times 2}$ is the identity matrix on $\mathbb{C}^2$.\\

\noindent Consequently, the collection of the maximal contexts $\Sigma = \{\Sigma^{(z)}, \Sigma^{(x)}, \Sigma^{(y)}\}$ contains all the projection operators on $\mathbb{C}^2$. As $L(\mathbb{C}^2)$ is the span of all such operators, $\Sigma = L(\mathbb{C}^2)$.\\

\noindent By Burnside's Theorem it must be then\smallskip

\begin{equation} \label{41} %{Eq.41}
   \mathrm{Lat}(\Sigma)
   =
   \mathrm{Lat}(\Sigma^{(z)})
   \cap
   \mathrm{Lat}(\Sigma^{(x)})
   \cap
   \mathrm{Lat}(\Sigma^{(y)})
   =
   \left\{
      \{0\}
      ,
      \mathbb{C}^2
   \right\}
   \;\;\;\;  ,
\end{equation}
\smallskip

\noindent which implies that the bivaluation $v_{|\Psi\rangle}\!: \{\Sigma^{(q)}\} \rightarrow \{0,1\}$ cannot be a total function.\\

\section{Concluding remarks}  %{5}

\noindent As it has been just shown, the Kochen-Specker theorem is the consequence of Burnside's theorem on the algebra of linear transformations on $\mathcal{H}$.\\

\noindent Indeed, according to the Kochen-Specker theorem, in a Hilbert space $\mathcal{H}$ of a finite dimension (greater than 3), it is impossible to assign to every projection operator in a set $\Sigma^{\prime}$ one of its eigenvalues, i.e., 1 or 0, in such a way that for any admissible state of the system $|\Psi\rangle$ the values $v_{|\Psi\rangle}(\hat{P}_{i}^{(w)})$ assigned to members $\hat{P}_{i}^{(w)}$ of a maximal context $\Sigma^{(w)} \subset \Sigma^{\prime}$ resolve to 1, that is, $\sum_{i}v_{|\Psi\rangle}(\hat{P}_{i}^{(w)}) = 1$.\\

\noindent On the other hand, the bivaluation of a projection operator is associated with the existence of its two nontrivial invariant subspaces. So, the inability to assign binary values, 1 or 0, to each projection operator in the set $\Sigma^{\prime}$ is the consequence of the fact that the family of subspaces $\mathrm{Lat}(\Sigma^{\prime})$ invariant under each maximal context in the set $\Sigma^{\prime}$ is irreducible, i.e., has no nontrivial invariant subspace.\\

\noindent In this way, contextuality (as an impossibility of assigning binary values to projection operators independently of their maximal contexts) is merely an inference from the fundamental theorem of noncommutative algebra, i.e., Burnside's Theorem.\\

\noindent It is worth mentioning that this theorem fails for finite dimensional vector spaces over the reals \cite{Shapiro}, such as the classical phase space $\Gamma$. This can be regarded as the algebraic reason for bivalence of classical mechanics.\\

\bibliographystyle{References}

\end{document}